\documentclass[conference]{IEEEtran}
\usepackage{mathpazo}
\usepackage{times}

\usepackage{amsmath}
\usepackage{amsfonts}
\usepackage{latexsym}
\usepackage{amssymb}
\usepackage{upref}
\usepackage{theorem}
\usepackage{graphicx}
\usepackage{psfrag}
\usepackage{cite}
\usepackage{epsfig}
\usepackage{color}
\usepackage{graphicx}
\usepackage{colortbl}

\newcommand{\tab}{\hspace*{1em}}

\theoremstyle{plain}
\theorembodyfont{\normalfont\slshape}

\newtheorem{thm}{Theorem$\!$}

\newtheorem{clm}[thm]{Claim$\!$}

\newtheorem{lem}[thm]{Lemma$\!$}

\newtheorem{prop}[thm]{Proposition$\!$}

\newtheorem{cor}[thm]{Corollary$\!$}

\newtheorem{defn}[thm]{Definition$\!$}

\newtheorem{xmpl}{Example$\!$}

\newtheorem{cnstr}{Construction$\!$}

\newtheorem{alg}{Algorithm$\!$}

\newcounter{enumrom}
\renewcommand{\theenumrom}{(\roman{enumrom})}

\makeatletter
\renewcommand{\@endtheorem}{\endtrivlist}
\makeatother

\makeatletter
\renewcommand{\thefigure}{{\@arabic\c@figure}}
\renewcommand{\fnum@figure}{{\bf Figure\,\thefigure}}
\makeatother

\newcommand{\cC}{{\cal C}}

\newcommand{\cM}{{\cal M}}

\DeclareMathOperator{\spun}{span}

\begin{document}


\title{\textbf{MDS Array Codes with Optimal Rebuilding}
\vspace*{-0.2ex}}

\author{\IEEEauthorblockN{Itzhak Tamo\IEEEauthorrefmark{1}\IEEEauthorrefmark{2}, Zhiying Wang\IEEEauthorrefmark{1} and Jehoshua Bruck\IEEEauthorrefmark{1}}
\IEEEauthorblockA{\IEEEauthorrefmark{1}Electrical Engineering Department,
California Institute of Technology,
Pasadena, CA 91125, USA \\}
\IEEEauthorblockA{\IEEEauthorrefmark{2}Electrical and Computer Engineering,
Ben-Gurion University of the Negev,
Beer Sheva 84105, Israel\\}
{\it \{tamo, zhiying, bruck\}@caltech.edu}\vspace*{-2.0ex}}

\maketitle

\begin{abstract}
MDS array codes are widely used in storage systems to protect data against erasures. We address the \emph{rebuilding ratio} problem, namely, in the case of erasures, what is the the fraction of the remaining information that needs to be accessed in order to rebuild \emph{exactly} the lost information? It is clear that when the number of erasures equals the maximum number of erasures that an MDS code can correct then the rebuilding ratio is $1$ (access all the remaining information). However, the interesting (and more practical) case is when the number of erasures is smaller than the erasure correcting capability of the code. For  example, consider an MDS code that can correct two erasures: What is the smallest amount of information that one needs to access in order to correct a single erasure? Previous work showed that the rebuilding ratio is bounded between $\frac{1}{2}$ and $\frac{3}{4}$, however, the exact value was left as an open problem. In this paper, we solve this open problem and prove that for the case of a single erasure with a $2$-erasure correcting code, the rebuilding ratio is $\frac{1}{2}$. In general, we construct a new family of $r$-erasure correcting MDS array codes that has optimal rebuilding ratio of $\frac{1}{r}$ in the case of a single erasure. Our array codes have efficient encoding and decoding algorithms (for the case $r=2$ they use a finite field of size $3$) and an optimal update property.
\end{abstract}



\section{Introduction}
Erasure-correcting codes are the basis of the ubiquitous RAID schemes for storage systems, where disks correspond to symbols in the code. Specifically, RAID schemes are based on MDS (maximum distance separable) array codes that enable optimal storage and efficient encoding and decoding algorithms. With $r$ redundancy symbols, an MDS code is able to reconstruct the original information if no more than $r$ symbols are erased. An array code is a two dimensional array, where each column corresponds to a symbol in the code and is stored in a disk in the RAID scheme. We are going to refer to a disk/symbol as a node or a column interchangeably, and an entry in the array as an element. Examples of MDS array codes are EVENODD \cite{Shuki-evenodd,Blaum96mdsarray}, B-code \cite{B-code}, X-code \cite{ x-code}, RDP \cite{RDP-code}, and STAR-code \cite{star-code}.

Suppose that some nodes are erased in an MDS array code, we will rebuild them by accessing (reading) some information in the surviving nodes. The fraction of the accessed information in the surviving nodes is called the \emph{rebuilding  ratio}. If $r$ nodes are erased, then the rebuilding ratio is $1$ since we need to read all the remaining information. However, is it possible to lower this ratio for less than $r$ erasures? For example, Figure \ref{fig:firstFigure} shows the rebuilding of the first \emph{systematic} (information) node for an MDS code with $4$ information elements and $2$ redundancy nodes, which requires the transmission of 3 elements. Thus the rebuilding ratio is $1/2$.

\begin{figure}[htp]
	\centering
	\includegraphics[width=.47\textwidth]{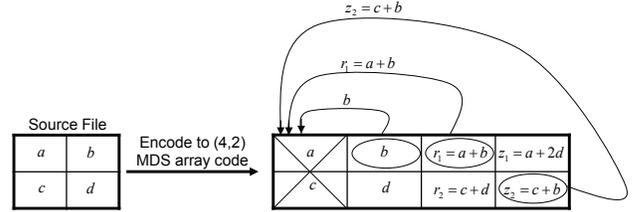}
\caption{Rebuilding of a $(4,2)$ MDS array code over $\mathbb{F}_{3}$. Assume the first node (column) is erased. }
	\label{fig:firstFigure}
	\vspace{-0.5cm}
	\end{figure}

In \cite{ Dimakis2010,Wu07deterministicregenerating}, a related problem is discussed: the nodes are assumed to be distributed and fully connected in a network, and the \emph{repair bandwidth} is defined as the minimum amount of data needed to transmit in the network in order to retain the MDS property.  Note that one block of data transmitted can be a function of several blocks of data. In addition, retaining MDS property does not imply rebuilding the original erased node, whereas we restrict our problem to \emph{exact} rebuilding. Therefore, the repair bandwidth is a lower bound of the rebuilding ratio.

An $(n,k)$ MDS code has $n$ nodes in each codeword and contains $k$ nodes of information and $r=n-k$ nodes of redundancy. A lower bound for the repair bandwidth was shown as \cite{Dimakis2010}
\begin{equation} \label{eq:tradeoff}
\frac{\cM}{k}\cdot \frac{n-1}{n-k} \ ,
\end{equation}
where $\cM$ is the total amount of information. It can be verified that Figure \ref{fig:firstFigure} matches this lower bound. A number of works addressed the repair bandwidth problem \cite{Dimakis2010,5206008,Dimakis-interference-alignment,Kumar09,Suh-alignment,Kumar2009,Changho-Suh2010,Cadambe2010,Wu07deterministicregenerating,Rashmi11}, and it was shown by interference alignment in \cite{Suh-alignment, Kumar2009} that this bound is asymptotically achievable for exact repair. Instead of trying to construct MDS codes that can be easily rebuilt, a different approach \cite{zhiying10, XiangXLC10} was used by trying to find ways to rebuild existing families of MDS array codes. The ratio  of rebuilding a single systematic node was shown to be $\frac{3}{4}+o(1)$ for EVENODD or RDP codes. However, from the lower bound of \eqref{eq:tradeoff} the ratio is $1/2$.

The main contribution of this paper is the first explicit construction of systematic $(n,k)$ MDS array codes for any constant $r=n-k$, which achieves optimal rebuilding ratio of $\frac{1}{r}$. We call them \emph{intersecting zigzag sets codes (IZS codes)}. The parity symbols are constructed by linear combinations of a set of information symbols, such that each information symbol is contained exactly once in each parity node. These codes have a variety of advantages: 1) they are systematic codes, and it is easy to retrieve information;
2) they are have high code rate $k/n$, which is commonly required in storage systems; 3) the encoding and decoding of the codes can be easily implemented (for $r=2$, the code uses finite field of size 3); 4) they match the lower bound (\ref{eq:tradeoff}) when rebuilding a systematic node; 5) the rebuilding of a failed node requires simple computation and access to only $1/(n-k)$ of the data in each node (no linear combination of data); and 6) they have \emph{optimal update}, namely, when an information element is updated, only $n-k+1$ elements in the array need update.

The remainder of the paper is organized as follows. Section \ref{sec2} provides definitions and background on MDS array codes. Section \ref{sec3} presents the new constructions of $(k+2,k)$ MDS array codes with an optimal rebuilding ratio. Section \ref{code-duplication} introduces the concept of code duplication that enables the constructions of $(k+2,k)$ MDS array codes for an arbitrary number of columns. We discuss the size of the finite field needed for these constructions in Section \ref{section 5}. Decoding algorithms for erasures and errors are discussed in Section \ref{sec:dec}. Section \ref{generalization} provides generalizations of our MDS code constructions to an arbitrary number of parity columns. Finally, we provide concluding remarks in Section \ref{summary}.

%
%
\section{Definitions and Problem Settings} \label{sec2}
%
%

In the rest of the paper, we are going to use $[i,j]$ to denote $\{i,i+1,\dots,j\}$ for integers $i \le j$. And denote the complement of a subset $X\subseteq M$ as $\overline{X}=M\backslash X$. For a matrix $A$, $A^T$ denotes the transpose of $A$.

Let $A=(a_{i,j})$ be an array of size $p\times k$ over a finite field $\mathbb{F}$, each entry of which is an information element. We add to the array two parity columns and obtain an $(n=k+2,k)$ MDS code of array size $p \times n$. Each element in these parity columns is a linear combination of elements from $A$. More specifically, let the two parity columns be $C_k=(r_0,r_1,...,r_{p-1})^T$ and $C_{k+1}=(z_0,z_1...,z_{p-1})^T$. Then for $l \in [0,p-1]$, $r_l=\sum_{a\in R_l}\alpha_aa\text{ and } z_l=\sum_{a\in Z_l}\beta_aa$, for some subsets $R_l,Z_l$ of elements in $A$, and some coefficients $\{\alpha_a\},\{\beta_a\}\subseteq \mathbb{F}$.
We will call $R=\{R_0,R_1,...,R_{p-1}\}$ and $Z=\{Z_0,Z_1,...,Z_{p-1}\}$ the sets that generate the parity columns.

Since the code is a $(k+2,k)$ MDS code, each information element should appear at least once in each parity column $C_{k},C_{k+1}.$ We will assume that each information element in $A$ appears exactly once in each parity column, which implies optimal update for the code. As a result, we have the following theorem.

\begin{thm}
The sets $R$ (or $Z$) are partitions of $A$ into $p$ equally sized sets of size $k$.
\end{thm}
\begin{IEEEproof}
Each set $X\in R$ does not contain two entries of $A$ from the same column. W.l.o.g. assume $X$ contains two entries of $A$ from the first column, then we can not rebuild these two elements when the first column and the parity column $C_{k+1}$ are erased. Thus $X$ contains at most one entry from each column and then $|X|\leq k.$ However each element of $A$ appears exactly once in each parity column, thus if there is $|X|<k$, $X$ $\in$ $R$, there is $Y\in R$, with $|Y|>k$, which leads to a contradiction. Therefore, $|X|=k$ for all $X \in R$. As each information element appear exactly once in the first parity column, $R=\{R_0,\dots,R_{p-1}\}$ are partitions of $A$ into $p$ equally sized sets of size $k$. Similar proof holds for the sets $Z=\{Z_0,\dots,Z_{p-1}\}$.
\end{IEEEproof}

By the above theorem, for the $j$-th systematic column $(a_0,\dots,a_{p-1})^T$, its $p$ elements are contained in $p$ distinct sets $R_l$, $l \in [0,p-1]$. In other words, the membership of the $j$-th column's elements in the sets $\{R_l\}$ defines a permutation $g_j:[0,p-1] \to [0,p-1]$, such that  $g_j(i)=l$ iff $a_i \in R_{l}$. Similarly, we can define a permutation $f_j$ corresponding to the second parity column, where $f_j(i)=l$ iff $a_i \in Z_{l}$.
For example, Figure \ref{fig:shapes} shows a $(5,3)$ code. Each element in the parity column Z is a linear combination of elements with the same symbol. And each systematic column corresponds to a permutation of the four symbols.

\begin{figure}
	\centering
\begin{tabular}{|c|c|c|c|c|c|}
	\hline
	& 0 & 1 & 2 & R & Z \\
	\hline
	0 & \cellcolor[gray]{0.8}{$\clubsuit$} & $\spadesuit$ & \cellcolor[gray]{0.8}{$\heartsuit$} & \cellcolor[gray]{0.8}{ } & \cellcolor[gray]{0.8}{$\clubsuit$} \\
	\hline
	1 & \cellcolor[gray]{0.8}{$\heartsuit$} & $\diamondsuit$ & \cellcolor[gray]{0.8}{$\clubsuit$} & \cellcolor[gray]{0.8}{ }& \cellcolor[gray]{0.8}{$\heartsuit$} \\
	\hline
	2 & $\spadesuit$ & $\clubsuit$ & $\diamondsuit$ & & $\spadesuit$ \\
	\hline
	3 & $\diamondsuit$ & $\heartsuit$ & $\spadesuit$ & & $\diamondsuit$ \\
	\hline
\end{tabular}
\caption{Permutations for zigzag sets in a $(5,3)$ code with $4$ rows. The permutations for rows are the identity permutations. The shaded elements are accessed to rebuild column 1.}
\label{fig:shapes}
\end{figure}

Observing that there is no importance of the elements' ordering in each column, w.l.o.g. we can assume that the first parity column contains the sum of each row of $A$ and $g_j$'s correspond to identity permutations, i.e. $r_i=\sum_{j=0}^{k-1} \alpha_{i,j}a_{i,j}$. We refer to the first and second parity columns as the row column and the zigzag column respectively, likewise $R_l$ and $Z_l$, $l \in [0,p-1]$, are referred to as row sets and zigzag sets respectively. Call $f_j$, $j \in [0,k-1]$, zigzag permutations. By assuming that the first parity column contains the row sums, the code is uniquely defined by (i) the zigzag permutations, and (ii) the coefficients in the linear combinations.

Our approach consists of two steps: first we choose the appropriate zigzag sets $Z_0,...Z_{p-1}$ in order to minimize the rebuilding ratio, and then we choose the coefficients in the linear combinations in order to make sure that the constructed code is indeed a $(k+2,k)$ MDS code.
But first we show that for any set of zigzag sets $Z=\{Z_0,...,Z_{p-1}\}$ there exists a $(k+2,k)$ MDS array code over a field $\mathbb{F}$ large enough. For that proof we use the well known Combinatorial Nullstellensatz by Alon \cite{Alon-polynomial-method}:
\begin{thm} (Combinatorial Nullstellensatz) \cite[Th 1.2]{Alon-polynomial-method}
\label{polynomial-method}
Let $\mathbb{F}$ be an arbitrary field, and let $f=f(x_1,...,x_q)$ be a polynomial in $\mathbb{F}[x_1,...,x_q]$. Suppose the degree of $f$ is $\deg(f)=\sum_{i=1}^q t_i$, where each $t_i$ is a nonnegative integer, and suppose the coefficient of $\prod_{i=1}^q x_i^{t_i}$ in $f$ is nonzero. Then, if $S_1,...,S_n$ are subsets of $\mathbb{F}$ with $|S_i| > t_i$, there are
$s_1\in S_1,s_2\in S_2,...,s_q\in S_q$ so that $$f(s_1,...,s_q)\neq 0.$$
\end{thm}

\begin{thm}
\label{zigzag-sets}
Let $A=(a_{i,j})$ be an array of size $p\times k$ and the zigzag sets be $Z=\{Z_0,...,Z_{p-1}\}$, then  there exists a $(k+2,k)$ MDS array code for $A$ with $Z$ as its zigzag sets over the field $\mathbb{F}$ of size greater than $p(k-1)+1$.
\end{thm}

\begin{IEEEproof}
Assume the information of $A$ is given in a column vector $W$ of length $pk$, where column $i$ of $A$ is in the row set $[(i-1)p+1,ip]$ of $W$. Each systematic node $i$, $i \in [0,k-1]$, can be represented as $Q_iW$ where $Q_i=[0_{p\times pi},I_{p\times p},0_{p\times p(k-i-1)}]$. Moreover define $Q_{k}=[I_{p\times p}, I_{p\times p},...,I_{p\times p}],Q_{k+1}=[x_0P_0,x_1P_1,...,x_{k-1}P_{k-1}]$ where the $P_i$'s are permutation matrices  of size $p\times p$, and the $x_i$'s are indeterminates, such that $C_{k}=Q_{k}W,C_{k+1}=Q_{k+1}W$. The permutation matrix  $P_i=(p^{(i)}_{l,m})$ is defined as $p^{(i)}_{l,m}=1$ if and only if $a_{m,i}\in Z_l$ and the $P_i$'s are not necessarily distinct. If there exists such
MDS code it is equivalent to the existence of a set of values for $\{x_i\}$ such that for any set of integers $\{s_1,s_2,...,s_k\}\subseteq [0,k+1]$ the matrix $Q=[Q_{s_1}^T,Q_{s_1}^T,...,Q_{s_k}^T ]$ is of full rank. It is easy to see that if the parity column $C_{k+1}$ is erased i.e., $k+1 \notin \{s_1,s_2,...,s_k\}$ then $Q$ is of full rank. If $k\notin \{s_1,s_2,...,s_k\}\text{ and } k+1\in \{s_1,s_2,...,s_q\}$ then $Q$ is of full rank if none of the $x_i$'s equals to zero. The last case is when both $k,k+1\in \{s_1,s_2,...,s_k\}$, i.e.,  there are $0\leq i<j\leq k-1$ such that $i,j\notin \{s_1,s_2,...,s_k\}.$ It is easy to see that in that case $Q$ is of full rank if and only if the submatrix
\[B_{i,j}= \left( \begin{array}{cc}
x_iP_i &x_jP_j  \\
I_{p\times p} & I_{p\times p}  \end{array} \right)\]
is of full rank. This is equivalent to $ \det(B_{i,j})\neq 0$. Note that $\deg( \det(B_{i,j}))=p$ and the coefficient of $x_i^p$ is $\det(P_i) \in \{1,-1\}$.
Define the polynomial $$T=T(x_0,x_1,...,x_{k-1})=\prod_{0\leq i<j\leq k-1}\det(B_{i,j}),$$ then the result follows if and only if there is an assignment $a_0,a_1,..,a_{k-1}\in \mathbb{F}$ such that $T(a_0,a_1,...,a_{k-1})\neq 0.$ $T$ is of degree $p\binom{k}{2}$ and the coefficient of
$\prod_{i=0}^{k-1} x_i^{p(k-i)}$ is $\prod_{i=0}^{k-1} \det(P_i)^{k-i}\neq 0$. Set for any $i, S_i=\mathbb{F} \backslash 0$ in Theorem \ref{polynomial-method}, then the result follows.
\end{IEEEproof}

The above theorem states that there exist coefficients such that the code is MDS, and thus we will focus on finding proper permutations $\{f_j\}$ first. The idea behind choosing the zigzag sets is as follows: assume a systematic column $(a_0,a_1,...,a_{p-1})^T$ is erased. Each element $a_i$  is contained in exactly one row set and one zigzag set. For rebuilding of element $a_i$, access the parity of its row set or zigzag set. Moreover access the values of the remaining elements in that set, except $a_i$. We say that an element $a_i$ is rebuilt by a row (zigzag) if the parity of its row set (zigzag set) is accessed. For example, in Figure \ref{fig:shapes} supposing column $1$ is erased, one can access the shaded elements and rebuild its first two elements by rows, and the rest by zigzags.

In order to minimize the number of accesses to rebuild the node, we need to minimize  the size of \begin{equation}
\label{eq:77}
|\cup_{i=0}^{k-1} S_i|
\end{equation}
where each $S_i \in R\cup Z$ is either a row set or a zigzag set containing $a_i$. Each $|S_i|=k$, therefore in order to minimize the size of the union we will try to maximize the number of intersections between the sets $\{S_i\}_{i=0}^{k-1}.$ We say that $S=(S_0,S_1,...,S_{k-1})$ rebuilds $(a_1,a_2,...,a_p)^T$. For the rebuilding of node $i$ by $S=(S_0,S_1,...,S_{k-1})$, define the number of intersections by $I(i|S)=pk-|\cup_{j=0}^{k-1}S_j|$. Moreover define the number of total intersections in an MDS array code with zigzag sets $Z$ as $$I(Z)=\sum_{i=0}^{k-1}\max_{S\text{ rebuilds } i}I(i|S).$$ Now define $h(k)$ to be the maximum possible intersections over all $(k+2,k)$ MDS array codes, i.e., $$h(k)=\max_{Z}I(Z)$$

In Figure \ref{fig:shapes} the rebuilding set is $S=\{R_0,R_1,Z_0,Z_1\}$, the size of (\ref{eq:77}) is $8$ and $I(1|S)=4$. Note that each surviving node accesses exactly $\frac{1}{2}$ of its information without performing any calculation within it. The following lemma gives a recursive bound for the number of intersections.

\begin{thm}
\label{th:inequality}
Let $q\leq k\leq p$ then $h(k)\leq \frac{k(k-1)h(q)}{q(q-1)}$
\end{thm}
\begin{IEEEproof}
Let $A$ be an information array of size $p \times k$. Construct an MDS array code $\cC$ such that the first parity is the row sums, and the second parity is defined by the zigzag sets $Z$. Suppose the zigzag sets $Z$ are defined such that $I(Z)$ is maximized.

Let $B=(b_{i,j})$ be a matrix of size $\binom{k}{q}\times k$ which is defined as follows: each row $j$ in the matrix corresponds to a subset of columns of size $q$ of $A$, denoted by $J$. Denote by the subarray of $A$ with columns in $J$ as $A_J$. Define zigzag sets $Z_J$ as subsets of $Z$ with only elements in columns $J$. For each $A_J$, construct an MDS array code $\cC_J$ using $A_J$ as information, row sum as the first parity, and $Z_J$ as the second parity.

Each column $i$ in $B$ corresponds to column $i$ in $A$. Let $b_{j,i}$ be the number of maximal intersections when rebuilding column $i$ in the code $\cC_J$.
If $i\notin J$ then $b_{j,i}=0.$ It is easy to see that the sum of each row is no more than $h(q)$.

On the other hand, the sum of the columns can be computed as follows. for the code $\cC$, denote by $I(l,i|S)$ the number of intersections in column $l$ for rebuilding column $i$ using the rebuilding set $S$.
Suppose $S_i^*$ rebuilds column $i$ and maximizes $I(i|S)$, then
$$\sum_{i}\sum_{l \neq i}I(l,i|S_i^*)=h(k).$$
Similarly, for the code $\cC_J$, $l \neq i \in J$, denote by $I_{J}(l,i|S)$ the number of intersections in column $l$ for rebuilding column $i$ using a rebuilding set $S$.
Assume $S_{i,J}^*$ rebuilds column $i$ in $\cC_J$ and maximizes $\sum_{l \neq i \in J}I_J(l,i|S)$, then by definition
$$b_{j,i}=\sum_{l \neq i \in J}I_J(l,i|S_{i,J}^*).$$
Note that $S_i^*$ also rebuilds column $i$ in $\cC_J$, but may not be the maximum-achieving set:
$$\sum_{l \neq i \in J}I_J(l,i|S_{i,J}^*) \ge \sum_{l \neq i \in J}I(l,i|S_i^*).$$
Combining  the above equations, we get
\begin{eqnarray*}
\sum_{i}\sum_{j}b_{j,i}&=&\sum_{i}\sum_{J}\sum_{l \neq i \in J}I_J(l,i|S_{i,J}^*) \\
& \ge & \sum_{i}\sum_{J}\sum_{l\neq i \in J} I(l,i|S_i^*) \\
&=& \sum_{i}\sum_{l \neq i} \sum_{J:l,i \in J} I(l,i|S_i^*) \\
&=& \sum_{i}\sum_{l \neq i} \binom{k-2}{q-2} I(l,i|S_i^*) \\
&=& \binom{k-2}{q-2} \sum_{i}\sum_{l \neq i} I(l,i|S_i^*) \\
&=&  \binom{k-2}{q-2} h(k)
\end{eqnarray*}
Thus the summation of the entire matrix $B$ leads us to $$h(k)\leq \frac{\binom{k}{q}h(q)}{\binom{k-2}{q-2}}=\frac{k(k-1)h(q)}{q(q-1)}$$
and completes the proof.
\end{IEEEproof}

For a fixed number of rows $p$ define the \emph{rebuilding ratio} for a $(k+2,k)$ MDS code $\cC$ as
$$R(\mathcal{C}) = \frac{\sum_{i=0}^{k-1} \textrm{accesses to rebuild node } i}{pk(k+1)},$$
which denotes the average fraction of accesses in the surviving array for rebuilding one systematic node.
Define the \emph{ratio function} for all $(k+2,k)$ MDS codes with $p$ rows as $$R(k)=\min_{\cC}{R(\cC)}=1-\frac{h(q)+pq}{pq(q+1)}.$$
which is the optimal average portion of the array needed to be accessed in order to rebuild one lost column. In this expression $h(q)$ and $pq$ in the numerator correspond to the number of elements we do not access in the systematic nodes and parity nodes, respectively.

\begin{lem}
\label{monotone function}
$R(k)$ is no less than $\frac{1}{2}$ and is a monotone nondecreasing function.
\end{lem}

\begin{IEEEproof}
From Theorem \ref{th:inequality} for $k=q+1$ we get $h(q+1)\leq \frac{(q+1)qh(q)}{q(q-1)}=\frac{(q+1)h(q)}{q-1}$. Then we get
\begin{eqnarray*}
R(q+1)&=&1-\frac{h(q+1)}{p(q+1)(q+2)}-\frac{1}{q+2} \\
&\geq & 1-\frac{h(q)}{p(q-1)(q+2)}-\frac{1}{q+2}
\end{eqnarray*}
We want to show
$$R(q+1) \ge R(q)$$
or
$$\frac{h(q)+p(q-1)}{p(q-1)(q+2)} \le \frac{h(q)+pq}{pq(q+1)}$$
which is equivalent to
\begin{equation} \label{half}
h(q) \le \frac{pq(q-1)}{2}.
\end{equation}
In the rebuilding of a node $(a_1,a_2,...,a_p)^T$, for any element $a_i$, one element from the parity nodes is accessed. In total we access in the two parity nodes $p$ elements out of $2p$ elements, i.e., exactly $\frac{1}{2}$ of the information. Let $x\text{ and } p-x$ be the number of elements that are accessed from the first and second parity respectively. W.l.o.g we can assume that $x\geq \frac{p}{2}$, otherwise $p-x$ would satisfy it. Each element of these $x$ sets is a sum of a set of size $q$. Thus in order to rebuild the node, we need to access at least $x(q-1)\geq \frac{p(q-1)}{2}$ elements in the $q-1$ surviving systematic nodes, which is at least half of the size of the surviving systematic nodes. So the number of intersections is no more than $\frac{pq(q-1)}{2}$.
Thus (\ref{half}) holds and the ratio function $R(k) \ge \frac{1}{2}$.
\end{IEEEproof}

The lower bound of $1/2$ in the above theorem can be also derived from \eqref{eq:tradeoff}.
We will see from Lemma \ref{lem13} that $R(k)$ is almost $1/2$ for all $k$ and  $p=2^l$, where $l$ is large enough.

By Lemma \ref{monotone function} and \eqref{eq:tradeoff} for any $p$ and $k$, $R(k)\geq \frac{1}{2}.$ For example, it can be verified that for the code in  Figure \ref{fig:shapes}, all the three systematic columns can be rebuilt by accessing half the remaining elements. Thus the code achieves the lower bound $1/2$, and therefore $R(3)=1/2$.


%
%
\section{$(k+2,k)$ MDS array code constructions} \label{sec3}
%
%

The previous section gave us a lower bound for the ratio function. The question is can we achieve it? If so, how? We know that each $(k+2,k)$ MDS array code with row and zigzag columns is defined by a set of permutations $f_0,...,f_{k-1}$ and their subsets $X_i$'s. The following construction constructs a family of such MDS array codes.
From any set $T\subseteq \mathbb{F}_2^m$, $|T|=k$, we construct a $(k+2,k)$ MDS array code of size $2^m\times (k+2).$
The ratio of the constructed code will be proportional to the size of the union of the rebuilding set.
Thus we will try to construct such permutations and subsets that minimize the union. We will show that some of these codes have the optimal ratio of $\frac{1}{2}$.

In this section all the calculations are done over $\mathbb{F}_2$. By abuse of notation we use $x\in [0,2^{m}-1]$ both to represent the integer and its binary representation. It will be clear from the context which meaning is in use.

\begin{cnstr}
Let $A=(a_{i,j})$ be an array of size $2^m\times k$ for some integers $k,m$ and $k\leq 2^m$. Let $T\subseteq \mathbb{F}^m_2$ be a subset of vectors of size $k$. For $v\in T$ we define the permutation $f_v:[0,2^m-1]\to [0,2^m-1] $ by $f_v(x)=x+v$, where $x$ is represented in its binary representation. One can check that this is actually a permutation. For example when $m=2,v=(1,0)$ $$f_{(1,0)}(3)=(1,1)+(1,0)=(0,1)=1,$$
and the corresponding permutation of $v$ is $[2, 3, 0, 1]$.
In addition, we define $X_v$ as the set of integers $x$ in $[0,2^m-1]$ such that the  inner product between their binary representation and $v$ satisfies $x\cdot v=0$, e.g., $X_{(1,0)}=\{0,1\}$. The construction of the two parity columns is as follows:
The first parity column is simply the row sums. The zigzag sets $Z_0,...,Z_{2^m-1} $are defined by the permutations $\{f_{v_j}:v_j\in T\}$ as $a_{i,j}\in Z_l$ if $f_{v_j}(i)=l.$ We will denote the permutation $f_{v_j}$ as $f_j$. Assume column $j$ needs to be rebuilt, and denote $S_r=\{a_{i,j}:i\in X_j\}$ and $S_z=\{a_{i,j}:i\notin X_j\}.$ Rebuild the elements in $S_r$ by rows and the elements in $S_z$ by zigzags.
\label{cnstr1}
\end{cnstr}

Recall that by Theorem \ref{zigzag-sets} this code can be an MDS code over a field large enough.
The following theorem gives the ratio for Construction \ref{cnstr1}.
\begin{thm}
\label{th:123}
The code described in Construction \ref{cnstr1} and generated by the vectors $v_0,v_1,...,v_{k-1}$ is a $(k+2,k)$  MDS array code with ratio
\begin{equation}
R=\frac{1}{2}+\frac{\sum_{i=0}^{k-1}\sum_{j\neq i}|f_i(X_i)\cap f_j(X_i)|}{2^mk(k+1)}.\label{eq:345}
\end{equation}
\end{thm}

\begin{IEEEproof}
In rebuilding of node $i$ we rebuild the elements in rows $X_i$ by rows, thus the first parity column $C_{k}$ accesses the values of the sum of rows $X_i$. Moreover, each surviving systematic node accesses its elements in rows $X_i$. Hence, by now $|X_i|k=2^{m-1}k$ elements are accessed, and we manage to rebuild the elements of node $i$ in rows $X_i$. The elements of node $i$ in rows $\overline{X_i}$ are rebuilt by zigzag, thus the second parity column $C_{k+1}$ accesses the values of zigzags $\{z_{f_i(l)}:l\in \overline{X_i}\}$ and each surviving systematic node accesses the elements of these zigzags from its column, unless these elements are already included in the rebuilding by rows. The zigzag elements in $\{Z_{f_i(l)}:l\in \overline{X_i}\}$ of node $j$ are in rows $f_j^{-1}(f_i(\overline{X_i}))$, thus the extra elements node $j$ need to access are  in rows $f_j^{-1}(f_i(\overline{X_i}))\backslash X_i.$ But,
\begin{align*}
|f_j^{-1}(f_i(\overline{X_i}))\backslash X_i|&=|\overline{f_j^{-1}(f_i(X_i))}\cap \overline{X_i}|\\
&=|f_j^{-1}(f_i(X_i))\cap X_i|\\
&=|(f_i(X_i))\cap f_j(X_i)|,
\end{align*}
where we used the fact that $f_i,f_j$ are bijections, and $|X_i|=|\overline{X_i}|=2^{m-1}$.
Hence in rebuilding of node $i$ the number of elements to be accessed is $2^{m-1}(k+1)+\sum_{j\neq i}|(f_i(X_i))\cap f_j(X_i)|.$ The result follows by dividing by the size of the remaining array $2^m(k+1)$ and averaging over  all systematic nodes.
\end{IEEEproof}

The following lemma will help us to calculate the sum in \eqref{eq:345}, but first we associate to any vector $v=(v_1,...,v_m)\in \mathbb{F}_2^m$ the subset of integers $B_v\subseteq [m]$ where $i\in B_v$ if $v_i=1.$

\begin{lem}
\label{lemma 3}
for any $v,u\in T$
\begin{equation}
|f_v(X_v)\cap f_u(X_v)|=
\begin{cases}
|X_v|, & |B_v \backslash B_u|\mod 2 =0\\
0,     & |B_v \backslash B_u|\mod 2 =1.
\end{cases}
\end{equation}
\end{lem}

\begin{IEEEproof}
Consider the group $(\mathbb{F}_2^m,+)$. Recall that $f_v(X)=X+v=\{x+v:x\in X\}$.
The sets $f_v(X_v)=X_v+v$ and $f_u(X_v)=X_v+u$ are cosets of the subgroup $X_v=\{w\in \mathbb{F}_2^m:w\cdot v=0\}$, and they are either identical or disjoint. Moreover, they are identical iff $v-u\in X_v$, namely $(v-u)\cdot v= \sum_{i:v_i=1,u_i=0}1=0.$ However, $|B_v \backslash B_u|\mod 2 = \sum_{i:v_i=1,u_i=0}1$, and the result follows.
\end{IEEEproof}

This construction enables us to construct an MDS array code from any subset of vectors in $\mathbb{F}_2^m$. However, it is not clear which subset of vectors should be chosen. The following is an example of a code construction for a specific set of vectors.
\begin{xmpl} \label{xmpl1}
Let $T=\{v \in \mathbb{F}_2^m:\|v\|_1=3\}$ be the set of vectors with weight 3 and length $m$. Notice that $|T|=\binom{m}{3}$. Construct the code $\cC$ by $T$ according to Construction \ref{cnstr1}.
Given $v \in T$, $|\{u \in T: |B_v\backslash B_u|=3\}|= \binom{m-3}{3}$, which is the number of vectors with 1's in different positions as $v$. Similarly,
$|\{u \in T: |B_v\backslash B_u|=2\}|= 3\binom{m-3}{2}$ and
$|\{u \in T: |B_v\backslash B_u|=1\}|= 3(m-3)$. By Theorem \ref{th:123} and Lemma \ref{lemma 3}, for large $m$ the ratio is
$$\frac{1}{2}+\frac{2^{m-1}\binom{m}{3}3\binom{m-3}{2}}{2^m\binom{m}{3} (\binom{m}{3}+1)} \approx \frac{1}{2} + \frac{9}{2m}.$$
\end{xmpl}

Note that this code reaches the lower bound of the ratio as $m$ tends to infinity. In the following we will construct codes that reach the lower bound exactly.

Let $\{f_0,...,f_{k-1}\}$ be a set of permutations over the set $[0,2^{m}-1]$ with associated subsets $X_0,...,X_{k-1}\subseteq [0,2^{m}-1]$, where each $|X_i|=2^{m-1}$. We say that this set is a set of \emph{orthogonal permutations} if for any $i,j\in [0,k-1]$ $$\frac{|f_i(X_i)\cap f_j(X_i)|}{2^{m-1}}=\delta_{i,j},$$ where $\delta_{i,j}$ is the Kronecker delta. Let $\{e_i\}_{i=1}^m$ be the standard vector basis of $\mathbb{F}_2^m$ and set $e_0$ to be the zero vector. The following theorem constructs a set of orthogonal permutations of size $m+1$.
\begin{thm}
\label{orthogonal-permutations}
The permutations $f_0,...,f_m$ and sets $X_0,...,X_m$ constructed by the vectors $\{e_i\}_{i=0}^m$ and Construction \ref{cnstr1} where $X_0$ is modified to be $X_0=\{x\in F_2^m:x\cdot (1,1,...,1)=0\}$ is a set of orthogonal permutations. Moreover the $(m+3,m+1)$ MDS array code of array size $2^m\times (m+3)$ defined by these permutations has \emph{optimal} ratio of $\frac{1}{2}.$ Hence, $R(m+1)=\frac{1}{2}$.
\end{thm}

\begin{IEEEproof}
Since $|B_{e_i} \backslash B_{e_j}|=1$ for any $i\neq j\neq 0$ we get by lemma \ref{lemma 3} $$f_i(X_i)\cap f_j(X_i)=\emptyset.$$
Note that $$f_i((x_1,x_2,...x_m))=(x_1,x_2,...,x_{i-1},\overline{x_i},x_{i+1},...,x_m)$$ and $$X_i=\{(x_1,x_2,...,x_m)\in F_2^m:x_i=0\}$$ thus
$$f_i(X_i)\cap f_0(X_i)=\{x\in F_2^m:x_i=1\}\cap \{x\in F_2^m:x_i=0\}=\emptyset.$$
Moreover,
\begin{align*}
 &f_0(X_0)\cap f_i(X_0) \\
=& \{x\in F_2^m:\sum_{i=1}^m x_i=0\}\cap \{x\in F_2^m:\sum_{i=1}^m x_i=1\}\\
=&\emptyset.
\end{align*}
Hence the permutations $f_0,\dots,f_m$ are orthogonal permutations and the ratio is $1/2$ by Theorem \ref{th:123}.
\end{IEEEproof}

Actually this set of orthogonal permutations is optimal in size, as the following theorem suggests.
\begin{thm}\label{thm:size}
Let $F$ be an orthogonal set of  permutations over the integers $[0,2^m-1]$, then the size of $F$ is at most $m+1.$
\end{thm}
\begin{IEEEproof}
We will prove it by induction on $m$. For $m=0$ there is nothing to prove. Let $F=\{f_0,f_1,...,f_{k-1}\}$ be a set of orthogonal permutations over the set $[0,2^m-1].$ We only need to show that $|F|=k\leq m+1.$ It is trivial to see that for any $g,h\in S_{2^m}$ the set $hFg=\{hf_0g,hf_1g,...,hf_{k-1}g\}$ is also a set of orthogonal permutations with sets $g^{-1}(X_0),g^{-1}(X_1),...,g^{-1}(X_{k-1}).$ Thus w.l.o.g. we can assume that $f_0$ is the identity permutation and $X_0=[0,2^{m-1}-1]$. From the orthogonality we get that
$$\cup_{i=1}^{k-1} f_i(X_0)=\overline{X_0}=[2^{m-1},2^{m}-1].$$ Note that for any $i\neq 0,|X_i\cap X_0|=\frac{|X_0|}{2}=2^{m-2}.$ Assume the contrary, thus w.l.o.g we can assume that $|X_i\cap X_0|>2^{m-2}$, otherwise $|X_i\cap \overline{X_0}|>2^{m-2}.$ For  any $j\neq i\neq 0$ we get that
\begin{equation}
\label{eq:444}
f_j(X_i\cap X_0), f_i(X_i\cap X_0)\subseteq \overline{X_0},
\end{equation}
\begin{equation}
\label{eq:445}
|f_j(X_i\cap X_0)|=|f_i(X_i\cap X_0|)>2^{m-2}=\frac{|\overline{X_0}|}{2}.
\end{equation}
From equations (\ref{eq:444}) and (\ref{eq:445}) we conclude that $f_j(X_i\cap X_0)\cap f_i(X_i\cap X_0)\neq \emptyset$, which contradicts the orthogonality property. Define the set of  permutations $F^*=\{f_i^*\}_{i=1}^{k-1}$ over the set of integers $[0,2^{m-1}-1]$ by $f_i^*(x)=f_i(x)-2^{m-1}$, which is a set of orthogonal permutations with sets $\{X_i\cap X_0\}_{i=1}^{k-1}$. By induction $k-1 \leq m$ and the result follows.
\end{IEEEproof}

The above theorem implies that any $(k+2,k)$ systematic MDS array code of size $2^m\times (k+2)$ with optimal update and ratio $\frac{1}{2}$, satisfies $k\leq m+1$. Notice that the code in Theorem \ref{orthogonal-permutations} achieves the upper bound, i.e. $k=m+1$.

\begin{xmpl}
Let $A$ be an array of size $4\times 3$. We construct a $(5,3)$ MDS array code for $A$ as in Theorem \ref{orthogonal-permutations} that accesses  $\frac{1}{2}$ of the remaining information in the array to rebuild any systematic node (see Figure \ref{fig2}).
\begin{figure}[t]
	\centering
	\includegraphics[scale=0.47]{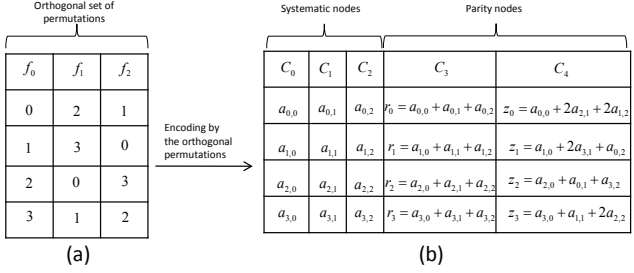}
\caption{$(a)$ The set of orthogonal permutations as in Theorem \ref{orthogonal-permutations} with sets $X_0=\{0,3\},X_1=\{0,1\},X_2=\{0,2\}$. $(b)$ A $(5,3)$ MDS array code generated by the orthogonal permutations. The first parity column $C_3$ is the row sum and the second parity column $C_4$ is generated by the zigzags. For example, zigzag $z_0$ contains the elements $a_{i,j}$ that satisfy $f_j(i)=0$.}
	\label{fig2}
\vspace{-0.5cm}
\end{figure}
For example, $X_1=\{0,1\}$, and for rebuilding  of node $1$ (column $C_1$) we  access the elements $a_{0,0},a_{0,2},a_{1,0},a_{1,2}$, and the following four parity elements
\begin{align*}
&r_0=a_{0,0}+a_{0,1}+a_{0,2}\\
&r_1=a_{1,0}+a_{1,1}+a_{1,2}\\
&z_{f_1(2)}=z_0=a_{0,0}+2a_{2,1}+2a_{1,2}\\
&z_{f_1(3)}=z_1=a_{1,0}+2a_{3,1}+a_{0,2}.
\end{align*}
It is trivial to rebuild node $1$ from the accessed information. Note that each of the surviving node accesses exactly $\frac{1}{2}$ of its elements. It can be easily verified that the other systematic nodes can be rebuilt the same way. Rebuilding a parity node is easily done by accessing all the information elements.
\end{xmpl}

%
%
\section{Code Duplication}
\label{code-duplication}
%
%
In this section, we are going to increase the number of columns in the constructed $(k+2,k)$ MDS codes with array size $2^m \times (k+2)$, such that $k > m+1$ and ratio is approximately $\frac{1}{2}$.

Let $\mathcal{C}$ be a $(k+2,k)$ array code where the zigzag sets $\{Z_l\}_{l=0}^{p-1}$ are defined by the set of permutations $\{f_i\}_{i=0}^{k-1} \subseteq S_p$ and $p$ is the number of rows in the array. For an integer $s$, an $s$-\emph{duplication code} $\mathcal{C}'$ is an $(sk+2,sk)$ MDS code with zigzag permutations defined by duplicating the $k$ permutations $s$ times each. Moreover, the first parity column is the row sums. The coefficients in the parities may be different from the code $\cC$. For an $s$-duplication code, denote the column corresponding to the  $t$-th $f_j$ as column $j^{(t)}$, $0 \le t \le s-1$. Call the columns $j^{(t)}$, $j \in [0,k-1]$, the $t$-th copy of the original code. An example of a $2$-duplication code of the code in Figure \ref{fig2} is illustrated in Figure \ref{fig:duplication}.

\begin{figure*}
	\centering
		\includegraphics[width=0.80\textwidth]{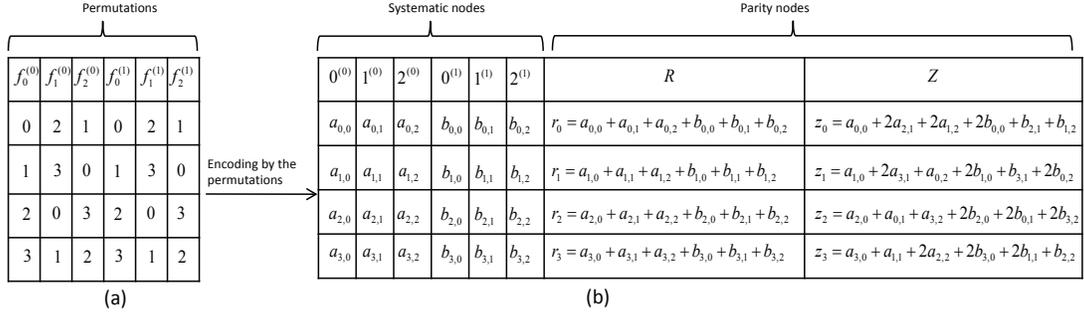}
	\caption{A $2$-duplication of the code in Figure \ref{fig2}. The code has $6$ information nodes and $2$ parity nodes. The ratio is $4/7$.}
	\label{fig:duplication}
\end{figure*}

\begin{thm} \label{lem13}
If a $(k+2,k)$ code $\mathcal{C}$ has ratio $R(\mathcal{C})$, then its $s$-duplication code $\mathcal{C}'$ has ratio at most $R(\mathcal{C})(1+\frac{s-1}{sk+1})$.
\end{thm}

\begin{IEEEproof}
Suppose in the rebuilding algorithm of $\mathcal{C}$, for column $i$,
elements of rows $J=\{j_1,j_2,\dots,j_u\}$ are rebuilt by zigzags, and the rest by rows. In $\mathcal{C}'$, all the $s$ columns corresponding to $f_i$ are rebuilt in the same way: the elements in rows $J$ are rebuilt by zigzags.
The rebuilding ratio of this algorithm is obviously an upper bound of $R(\cC')$.
W.l.o.g. assume column $i^{(0)}$ is erased, then the average (over all $i \in [0,k-1]$) of the number of elements needed to access in columns $l^{(t)}$, for all $l \neq i, l \in [0,k-1]$ is
$$R(\cC)p(k+1)-p,$$
for any $t \in [0,s-1]$. Here the term $-p$ corresponds to the access of the parity nodes in $\cC$. Moreover, we need to access all the elements in columns $i^{(t)}, 0<t \le s-1$, and access $p$ elements in the two parity columns. Therefore, the rebuilding ratio is
\begin{eqnarray*}
 R(\mathcal{C}')
&\le & \frac{s(R(\mathcal{C})p(k+1)-p)+(s-1)p+p}{p(sk+1)} \\
&=& R(\mathcal{C})\frac{s(k+1)}{sk+1} \\
&=& R(\mathcal{C})(1+\frac{s-1}{sk+1})
\end{eqnarray*}
\end{IEEEproof}

\begin{cor} \label{thm2}
The $s$-duplication of the code in Theorem \ref{orthogonal-permutations} has ratio $\frac{1}{2}(1+\frac{s-1}{s(m+1)+1})$, which is $\frac{1}{2}+\frac{1}{2(m+1)}$ for large $s$.
\end{cor}

When $s=1$, the code $\mathcal{C}$ in Theorem \ref{orthogonal-permutations} matches the lower bound of $1/2$, but when $s \ge 2$, the ratio is greater than the lower bound.

By Theorem \ref{thm2} we know that duplication of an optimal code has rebuilding ratio a little more than $1/2$. In fact, removing the first information column (corresponding to the identity $f_0$), the code in Theorem \ref{orthogonal-permutations} has $m$ information columns, $2^m$ rows, and the ratio is $1/2+1/(2m)$ for large $s$. In general, given $2^m$ rows, we can construct a $(k+2,k)$ code by $k$ vectors $\{v_0,\dots,v_{k-1}\}$ as in Construction \ref{cnstr1}, with some ratio $r$. And its $s$-duplication has ratio $r'=(1+\frac{1}{k})r$ for large $s$. Is it possible that $k$ is much larger than $m$ and $r$ is very close to $1/2$, such that $r'<1/2+1/(2m)$? For example, if we can find a code with $k=m^2$ information columns, and $r=1/2+1/(m^{1.5})$, then its $s$-duplication has ratio $r'= (1+1/m^2)r \approx 1/2+O(1/m^{1.5})$ for large $s$. And this code would have better ratio than Corollary \ref{thm2}. In the following, we are going to show that such codes do not exist. Actually the construction made by the standard vector basis is optimal.

Let $D=D(V,E)$ be a directed graph with set of vertices $V$ and directed edges $E$. Let $S$ and $T$ be two disjoint subsets of vertices, we define the density of the set $S$ to be $d_S=\frac{E_S}{|S|^2}$ and the density between $S$ and $T$ to be $d_{S,T}=\frac{E_{S,T}}{2|S||T|}$, where $E_S$ is the number of edges with both of its endpoints in $S$, and $E_{S,T}$ is the number of  edges incident with a vertex in $S$ and a vertex in $T$. The following technical lemma will help us to prove the optimality of our construction.

\begin{lem}
\label{density lemma}
Let $D=D(V,E)$ be a directed graph and $S, T$ be subsets of $V$, such that $S\cap T=\emptyset, S\cup T=V.$ We have: \\
(i) If $d_{S,T}<\max{\{d_S,d_T\}}$, then $d_V<\max{\{d_S,d_T\}}.$ \\
(ii) If $d_{S,T}\geq\max{\{d_S,d_T\}}$, then $d_V\leq d_{S,T}.$
\end{lem}

\begin{IEEEproof}
Note that $d_V=\frac{|S|^2d_S+|T|^2d_T+2|S||T|d_{S,T}}{|V|^2}.$ W.l.o.g assume that $d_S\geq d_T$ therefore
\begin{align*}
d_V&=\frac{|S|^2d_S+|T|^2d_T+2|S||T|d_{S,T}}{|V|^2}\\
&= \frac{|S|^2d_S+|T|^2d_S-|T|^2d_S+|T|^2d_T+2|S||T|d_{S}}{|V|^2}\\
&=\frac{d_S(|S|+|T|)^2-|T|^2(d_S-d_T)}{|V|^2}\\
&\leq d_S.
\end{align*}
If $d_{S,T}\geq\max{\{d_S,d_T\}}$ then
\begin{align*}
d_V&=\frac{|S|^2d_S+|T|^2d_T+2|S||T|d_{S,T}}{|V|^2}\\
&\leq \frac{|S|^2d_{S,T}+|T|^2d_{S,T}+2|S||T|d_{S,T}}{|V|^2}\\
&=d_{S,T}
\end{align*}
Thus the result follows.
\end{IEEEproof}

Define the directed graph $D_m=D_m(V,E)$ as $V=\{W:W\subseteq [m]\}$. We also view the vertices as binary vectors of length $n$ where each subset $W$ corresponds to its indicator vector. By abuse of notation, $W$ is used to represent both a subset of $[m]$ and a binary vector of length $m$.
There is a directed edge from $W_1$ to $W_2$ if and only if $|W_2 \backslash W_1|$ is odd.
Let $H$ be an induced subgraph of $D_m$. We construct the code $\cC(H)$ from the vetices of $H$ by Construction \ref{cnstr1}.
By Lemma \ref{lemma 3} we know an edge from $W_1$ to $W_2$ in $H$ means $f_{W_2}(X_{W_2}) \cap f_{W_1} (X_{W_2})= \emptyset$, so only half information from the column corresponding to $W_1$ is needed when we rebuild the column corresponding to $W_2$.
Then in an $s$-duplication of the code, when we rebuild a column in the $i$-th copy, the average ratio accessed in the $j$-th copy would be
$$\frac{|V(H)|^2 2^m-|E(H)|2^{m-1}}{|V(H)|^2 2^m} =
1-\frac{E(H)}{2|V(H)|^2}=1-\frac{d_H}{2}$$
for $i \neq j$, which equals the rebuilding ratio of the $s$-duplication code for large $s$. Namely,
\begin{equation}\label{eq:graph}
\lim_{s \to \infty}R(\cC')=1-\frac{d_H}{2}.
\end{equation}

The following theorem states that the ratio in Corollary \ref{thm2} of the $s$-duplication of the code from the standard basis is optimal among all codes constructed by binary vectors and duplication.

\begin{thm} \label{thm:opt rate}
For any induced subgraph $H$ of $D_m, d_H\leq \frac{m-1}{m}$.
Therefore, for fixed $m$ and large $s$, the ratio $\frac{1}{2}(1+\frac{1}{m})$  is optimal among all $s$-duplication codes constructed by binary vectors of length $m$ and Construction \ref{cnstr1}.
\end{thm}

\begin{IEEEproof}
We say that a binary vector is an even (or odd) vector if $\|W_2\|_1$ is 0 (or 1).
For two binary vectors $W_1,W_2$, $|W_2 \backslash W_1|$ being odd is equivalent to
$$1=W_2 \cdot \overline{W_1}= W_2(1+W_1)=\|W_2\|_1+W_2\cdot W_1$$
When $W_2$ is odd, this means $W_2W_1=0$ or $W_1,W_2$ are orthogonal vectors.
When $W_2$ is even, $W_1,W_2$ are not orthogonal.
Hence, one can check that when $W_1$ and $W_2$ are both odd (or even), there is either no edge or 2 edges between them.
Moreover, when $W_1$ is odd and $W_2$ is even, exactly one of the following is true: $|W_2 \backslash W_1|$ is odd, or $|W_1 \backslash W_2|$ is odd. Thus we have exactly one edge between $W_1$ and $W_2$.

Assume that there exist subgraphs of $D_m$ with density higher than $\frac{m-1}{m}$.
Let $H$ be such a subgraph with maximal density. Let $S$ and $T$ be the set of even and odd vectors of $H$ respectively. Note that $d_{S,T}=\frac{1}{2}$, because between any even and odd vertices there is exactly one directed edge. It is trivial to check that if $\max{\{d_S,s_T\}}\leq \frac{1}{2}$ then $d_H\leq \frac{1}{2}$ which leads to a contradiction. However if $\max{\{d_S,s_T\}}> \frac{1}{2}=d_{S,T}$ then by Lemma \ref{density lemma} we get a contradiction for the assumption that $H$ is the subgraph with the maximal density. Thus $H$ contains only odd vectors or even vectors. Let $V(H)=\{v_1,...,v_k\}$, and assume that the dimension of the subspace spanned by these vectors in $\mathbb{F}_2^m$ is $l$ where $v_1,v_2,...v_l$ are basis for it. Define $S=\{v_1,...v_l\},T=\{v_{l+1},...,v_k\}$. The following two cases show that the density can not be higher than $\frac{m-1}{m}$.

{\bf $H$ contains only odd vectors:}  Let $u\in T$. Since $u\in \spun(S)$ there is at least one $v\in S$ such that
$u\cdot v\neq 0$ and thus $(u,v),(v,u)\notin E(H)$, therefore the number of directed edges between $u$ and $S$ is at most $2(l-1)$ for all $u \in T$, which means $d_{S,T}\leq \frac{l-1}{l}\leq \frac{m-1}{m}$ and we get a contradiction by Lemma \ref{density lemma}.

{\bf $H$ contains only even vectors:} Since the $v_i$'s are even the dimension of $\spun(S)$ is at most $m-1$ (since for example $(1,0,...,0)\notin \spun (S)$) thus $l\leq m-1.$ Let $H^*$ be the induced subgraph of $D_{m+1}$ with vertices $V(H^*)=\{(1,v_i)|v_i\in V(H))\}$. It is easy to see that all the vectors of $H^*$ are odd, $((1,v_i),(1,v_j)\in E(H^*)\text{ if and only if } (v_i,v_j)\in E(H)$  and the dimension of $\spun(V(H^*))$ is at most $l+1 \le m.$ By the case already proven for odd vectors we conclude that
\begin{align*}
d_H=d_{H^*}&=\frac{\dim(\spun(V(H^*)))-1}{\dim(\spun(V(H^*)))}\\
&\leq\frac{l+1-1}{l+1}\\
&\leq\frac{m-1}{m},
\end{align*}
and get a contradiction by Lemma \ref{density lemma}.	

Therefore, the asymptotic optimal rebuilding ratio for duplication codes generated by binary vectors is
$\frac{1}{2}(1+\frac{1}{m})$ by \eqref{eq:graph}
\end{IEEEproof}

%
%
\section{Finite Field Size of a Code}\label{section 5}
%
%
We have already shown that the code $\mathcal{C}$ in Theorem \ref{orthogonal-permutations} has the best ratio, and if the code is over some big enough finite field $\mathbb{F}$, it is MDS. In this section, we will discuss in more detail about the field size such that the code $\mathcal{C}$ or its $s$-duplication $\mathcal{C}'$ is MDS.

Consider the code $\mathcal{C}$ constructed by Theorem \ref{orthogonal-permutations}.
Let the information in row $i$, column $j$ be $a_{i,j} \in \mathbb{F}$. Let its row coefficient be $\alpha_{i,j}$ and zigzag coefficient be $\beta_{i,j}$.
For a row set $R_u=\{a_{u,0},a_{u,1},\dots,a_{u,k}\}$, the row parity is $r_u=\sum_{j=0}^{k} \alpha_{u,j}a_{u,j}$. For a zigzag set $Z_u=\{a_{l_0, 0},a_{l_1,1},\dots,a_{l_k,k}\}$, the zigzag parity is $z_u = \sum_{j=0}^{k}\beta_{l_j,j}a_{l_j,j}$.

The $(m+3,m+1)$ code is MDS if we can recover the information when 1 or 2 columns are erased. It is easy to do rebuilding when 1 or 0 information column is erased. And $\alpha_{i,j}$, $\beta_{i,j}$ should be non-zero if one information column and one parity column are erased. For now, we will assume that $\alpha_{i,j}=1$ for all $i,j$.

\begin{cnstr} \label{cons3}
For the code $\mathcal{C}$ in Theorem \ref{orthogonal-permutations} over $\mathbb{F}_3$, define $u_i=\sum_{l=0}^i e_l$ for $0 \le i \le m$. Assign row coefficients as $\alpha_{i,j}=1$ for all $i,j$, and zigzag coefficients as
$$\beta_{i,j}= \begin{cases}
2, & \text{ if } u_j \cdot i =1 \\
1, &\text{ otherwise,}
\end{cases}$$
where $i=(i_1,\dots,i_m)$ is represented in binary and the operations are done over $\mathbb{F}_2$.
\end{cnstr}

\begin{thm} \label{thm6}
Construction \ref{cons3} is an $(m+3,m+1)$ MDS code.
\end{thm}
\begin{IEEEproof}
In an erasure of two systematic columns $i,j\in [0,m],i<j$, we access the entire remaining information in the array. Set $r'=r+e_i+e_j$, and recall that $a_{r,i}\in Z_{l}$ if $l=r+e_i$, thus $a_{r,i},a_{r',j}\in Z_{r+e_i}$ and $a_{r,j},a_{r',i}\in Z_{r+e_j}.$ From the two parity columns we need to solve the following equations (for some $y_1,y_2,y_3,y_4 \in \mathbb{F}_3$)
$$
\left[\begin{array}{cccc}
1 & 1 & 0 & 0 \\
0 & 0 & 1 & 1 \\
\beta_{r,i} & 0 & 0 & \beta_{r',j} \\
0 & \beta_{r,j} & \beta_{r',i} & 0
\end{array} \right]
\left[ \begin{array}{c}
a_{r,i} \\
a_{r,j} \\
a_{r',i} \\
a_{r',j} \\
\end{array} \right]
= \left[ \begin{array}{c}
y_1 \\
y_2 \\
y_3 \\
y_4
\end{array} \right].
$$
This set of equations is solvable if
\begin{equation} \label{eq1}
\beta_{r,i} \beta_{r',i} \neq \beta_{r,j}  \beta_{r',j}.
\end{equation}
For columns $0 \le i<j \le m$ and rows $r,r'$ defined above,
$$u_i \cdot r + u_i \cdot r' = u_i \cdot (r+ r') = \sum_{l=0}^{i}e_l (e_i+e_j) = e_j^2=1,$$
where the calculations  are done modulo 2. And
$$u_j \cdot r + u_j \cdot r' = u_j \cdot (r+ r') = \sum_{l=0}^{j}e_l (e_i+e_j) =  e_i^2+e_j^2=0$$
Thus $\beta_{r,i}\neq \beta_{r',i}$ and $\beta_{r,j}= \beta_{r',j}$. Note that each of the $\beta$'s is either $1$ or $2$, so \eqref{eq1} is satisfied and $\mathbb{F}_3$ ensures the code to be MDS.
\end{IEEEproof}

It is clear that $\mathbb{F}_2$ does not suffice for an MDS code, so $3$ is the optimal field size. The coefficients in Figure \ref{fig2} are assigned by Construction \ref{cons3}.

For $s$-duplication code $\mathcal{C}'$ in Theorem \ref{thm2}, denote the coefficients for the element in row $i$ and column $j^{(t)}$ by $\alpha_{i,j}^{(t)}$ and $\beta_{i,j}^{(t)}$, $0 \le t \le s-1$. Let $\mathbb{F}_q$ be a field of size $q$, and suppose its elements are $\{0,a^0,a^1,\dots,a^{q-2}\}$ for some primitive element $a$.

\begin{cnstr} \label{cons4}
For the $s$-duplication code $\mathcal{C}'$ in Theorem \ref{thm2} over $\mathbb{F}_q$, assign $\alpha_{i,j}^{(t)}=1$ for all $i,j,t$.
For odd $q$, let $s \le q-1$ and assign for all $t \in [0,s-1]$
$$\beta_{i,j}^{(t)}= \left\{
\begin{array}{ll}
a^{t+1}, & \text{if }u_j \cdot i=1 \\
a^{t},  & \text{o.w.}
\end{array}
\right.
$$
where $u_j=\sum_{l=0}^{j}e_l$. For even $q$ (power of 2), let $s \le q-2$ and assign for all $t \in [0,s-1]$
$$\beta_{i,j}^{(t)}= \left\{
\begin{array}{ll}
a^{-t-1}, & \text{if }u_j \cdot i=1 \\
a^{t+1},  & \text{o.w.}
\end{array}
\right.
$$
\end{cnstr}

Notice that the coefficients in each duplication has the same pattern as Construction \ref{cons3} except that values 1 and 2 are replaced by $a^t$ and $a^{t+1}$ if $q$ is odd (or $a^{t+1}$ and $a^{-t-1}$ if $q$ is even).

\begin{thm} \label{thm3}
Construction \ref{cons4} is an $(s(m+1)+2,s(m+1))$ MDS code.
\end{thm}

\begin{IEEEproof}
For the two elements in columns $i^{(t_1)},i^{(t_2)}$ and row $r$, $t_1 \neq t_2$, we can see that they are in the same row set and the same zigzag set. The corresponding two equations from the two parity elements are linearly independent if $\beta_{r,i}^{(t_1)} \neq \beta_{r,i}^{(t_2)}$, which is satisfied by the construction.

For the four elements in columns $i^{(t_1)},j^{(t_2)}$ and rows $r,r'=r+e_i+e_j$, $0 \le t_1, t_2 \le s-1$, $0 \le i < j \le m$, the code is MDS if
$$\beta_{r,i}^{(t_1)} \beta_{r',i}^{(t_1)} \neq \beta_{r,j}^{(t_2)} \beta_{r',j}^{(t_2)}.$$
By proof of Theorem \ref{thm3}, we know that  $\beta_{r,i}^{(t_1)}\neq \beta_{r',i}^{(t_1)}$, and $\beta_{r,j}^{(t_2)}= \beta_{r',j}^{(t_2)}=a^x$ for some $x$. When $q$ is odd,
$$\beta_{r,i}^{(t_1)} \beta_{r',i}^{(t_1)}=a^{t_1}a^{t_1+1}=a^{2t_1+1} \neq a^{2x}$$
for any $x$ and $t_1$. When $q$ is even,
$$\beta_{r,i}^{(t_1)} \beta_{r',i}^{(t_1)}=a^{t_1+1}a^{-t_1-1}=a^0 \neq a^{2x}$$ for any $t_1$ and $1 \le x \le q-2$ (mod $q-1$). And by construction, $x=t_2+1$ or $x=-t_2-1$ for $0 \le t_2 \le s-1 \le q-3$, so $1 \le x \le q-2$ (mod $q-1$).
Hence, the construction is MDS.
\end{IEEEproof}

\begin{thm}
For an MDS $s$-duplication code, we need a finite field $\mathbb{F}_q$ of size  $q \ge s+1$. Therefore, Theorem \ref{thm3} is optimal for odd $q$.
\end{thm}
\begin{IEEEproof}
Consider the two information elements in row $i$ and columns $j^{(t_1)},j^{(t_2)}$, which are in the same row and zigzag sets, for $t_1 \neq t_2 \in [0,s-1]$. The code is MDS only if
$$\left[ \begin{array}{cc}
\alpha_{i,j}^{(t_1)} & \alpha_{i,j}^{(t_2)} \\
\beta_{i,j}^{(t_1)} & \beta_{i,j}^{(t_2)}
\end{array} \right]$$
has full rank. All the coefficients are nonzero (consider erasing a parity column and a systematic column). Thus,
$(\alpha_{i,j}^{(t_1)})^{-1} \beta_{i,j}^{(t_1)} \neq (\alpha_{i,j}^{(t_2)})^{-1} \beta_{i,j}^{(t_2)}$, and $(\alpha_{i,j}^{(i)})^{-1} \beta_{i,j}^{(i)}$ are distinct nonzero elements in $\mathbb{F}_q$, for $i \in [0,s-1]$. So $q \ge s+1$.
\end{IEEEproof}

The coefficients in Figure \ref{fig:duplication} are assigned as Construction \ref{cons4} and $\mathbb{F}_3$ is used. One can check that any two column erases can be rebuilt in this code.

Consider for example of an $s$-duplication code with $m=10$ of the code in Theorem \ref{orthogonal-permutations}, the array is of size $1024 \times (11s+2)$. For $s=2$ and $s=6$, the ratio is $0.522$ and $0.537$ by Corollary \ref{thm2}, the code length is $24$ and $68$, and the field size needed can be $4$ and $8$ by Theorem \ref{thm3}, respectively. Both of these two codes are suitable for practical applications.

As noted before the optimal construction yields to a ratio of $1/2+1/m$ by using duplication of the code. However the field size is a linear function of the number of duplication of the code.
Is it possible to extend the number of columns in the code while using a constant field size? We know how to have $O(m^3)$ columns by using $O(m^2)$ duplications, however, the field size is $O(m^2)$.

The following code construction is a modification of Example \ref{xmpl1} and uses only the field of size $9$. The rebuilding ratio is asymptotically optimal, namely, $1/2+O(1/m)$.
Let the elements of $\mathbb{F}_9$ be $\{0,a^0,a^1,\dots,a^7\}$ where $a$ is a primitive element.

\begin{cnstr} \label{k/3}
Let $3|m$, and consider the following set of vectors $S\subset \mathbb{F}_2^m$: for each vector $v=(v_1,\dots,v_m) \in S$, $\|v\|_1=3$ and $v_i,v_j,v_l=1$ for some $i \in [1,m/3],j \in [m/3+1,2m/3], l \in [2m/3+1,m]$. For simplicity, we write $v=\{i,j,l\}$. Construct the $(k+2,k)$ code as in Construction \ref{cnstr1} using the set of vectors $S$, hence $k=|S|=(\frac{m}{3})^3=\frac{m^3}{27}$.
For $i\in [1,m/3]$, define $M_i = \sum_{t=1}^{i}e_t$. Similarly, for $j \in [m/3+1,2m/3]$, define $M_j=\sum_{t=m/3+1}^{j}e_t$ and for $l \in [2m/3+1,m]$, let $M_l=\sum_{t=2m/3+1}^{l}e_t$.
Then define a $3 \times m$ matrix
$$M_v=\left[
\begin{array}{lll}
M_i^T & M_j^T & M_l^T
\end{array}
\right]$$
 for $v=\{i,j,l\}$, and $M_i^T$ is the transpose of $M_i$.
Assign the row coefficients as $1$ and the zigzag coefficient for row $r$ column $c(v)$ as $a^t$, where $t=rM_v $ (in its binary expansion) and $a$ is the primitive element.
\end{cnstr}

For example, let $m=6$, and $v=\{1,4,6\}=(1,0,0,1,0,1) \in S$. The corresponding matrix is
$$M_v = \left[
\begin{array}{llllll}
1 & 0 & 0 & 0 & 0 & 0 \\
0 & 0 & 1 & 1 & 0 & 0 \\
0 & 0 & 0 & 0 & 1 & 1
\end{array}
\right]^T.
$$
For row $r=26=(0,1,1,0,1,0)$, we have
$$ t = rM_v  = (0,1,1)=3,$$
and the zigzag coefficient is $a^3$.

\begin{thm}
Construction \ref{k/3} is a $(k+2,k)$ MDS code with array size $2^m \times (k+2)$ and $k=m^3/27$. Moreover, the rebuilding ratio is $\frac{1}{2}+\frac{9}{2m}$ for large $m$.
\end{thm}
\begin{IEEEproof}
For each vector $v \in S$, there are $3(m/3-1)^2$ vectors $u\in S$ such that they have one $1$ in the same location as $v$, i.e. $|B_v\backslash B_u|=2$. Hence by Theorem \ref{th:123} and Lemma \ref{lemma 3}, for large $m$ the ratio is
$$\frac{1}{2}+\frac{3((\frac{m}{3})-1)^2}{2(\frac{m^3}{27}+1)} \approx \frac{1}{2} + \frac{9}{2m}$$.

Now we show that the code is indeed MDS. Consider columns $c(u),c(v)$ for some $u=\{i_1,j_1,l_1\} \neq v=\{i_2,j_2,l_2\}$ and $i_1, i_2 \in [1,m/3],j_1, j_2 \in [m/3+1,2m/3], l_1,l_2 \in [2m/3+1,m]$. And consider rows $r$ and  $r'=r+u+v$. The same as proof of Theorem \ref{thm6}, we know that the $4$ elements in these rows and columns are in two row sets and two zigzag sets. Moreover, we know from that proof that if $l_1 \neq l_2$, say $l_1 < l_2$, then
\begin{equation} \label{eq17}
rM_{l_1}^T \neq  r'M_{l_1}^T
\end{equation}
 and
\begin{equation} \label{eq18}
rM_{l_2}^T = r'M_{l_2}^T.
\end{equation}
If $l_1 = l_2$, then
\begin{equation} \label{eq19}
 rM_{l_1}^T = r'M_{l_1}^T = rM_{l_2}^T =r'M_{l_2}^T
\end{equation}
And similar results hold for $i_1,i_2$ and $j_1,j_2$.
Next we check if \eqref{eq1} is satisfied, which is equivalent to check
\begin{equation} \label{eq16}
rM_{u}^T +  r'M_{u}^T \neq rM_{v}^T +  r'M_{v}^T \mod 8
\end{equation}
by Construction \ref{k/3}. Here we view each vector of length $3$ as an integer in $[0,7]$ and the addition is ordinary addition of integers instead of bitwise XOR.

We do the addition in \eqref{eq16} from the lsb (least significant bit) to the msb (most significant bit). If the sum in a bit is more than $1$, the \emph{carry} is $1$ in the next bit. Notice that mod $8$ only changes the fourth bit and will not affect the first 3 bits.

If $l_1 \neq l_2$, the lsb in the sum in \eqref{eq16} are different on the two sides  by \eqref{eq17} \eqref{eq18}, and we are done. Assume $l_1 = l_2$, by \eqref{eq19} the lsb in the sum in \eqref{eq16} are $0$ on both sides, and the carry in the second bit are equal. In this case, if $j_1 \neq j_2$, since the carry are the same, by similar reason \eqref{eq16} holds. If $j_1=j_2$, again the second bit are equal and the carry in the third bit are equal. In the latter case $i_1 \neq i_2$ since $u \neq v$, and we have the third bit are different. Therefore, \eqref{eq16} is satisfied for any two vectors in $T$, and the code is MDS.
\end{IEEEproof}

Notice that if we do mod $15$ in \eqref{eq16} instead of mod $8$, the proof still follows because $15$ is greater than the largest possible sum in the equation. Therefore, a field of size $16$ is also sufficient to construct an MDS code, and it is easier to implement in a storage system.

Suppose $w<m$ and $w|m$. Construction \ref{k/3} can be easily generalized to the set of vectors $T$, such that any $v \in T$ satisfies $\|v\|_1=w$, $v_{i_j}=1$, for some $i_j \in [jm/w+1,(j+1)m/w]$, $j \in [0,\dots,w-1]$. The $(k+2,k)$ MDS code has array size $2^m \times (k+2)$ with
$$k=(\frac{m}{w})^w,$$
and uses a finite field of size
$$q = 2^w+1$$
if $2^w+1$ is a power of a prime. In all cases, $q=2^{w+1}$ suffices to make the code MDS. Moreover, when $w$ is odd and small compared to $m$, the ratio is
$$\frac{1}{2}+\frac{w^2}{2m}$$
for large $m$.

%
%
\section{Decoding of the Codes} \label{sec:dec}
%
%
In this section, we will discuss decoding algorithms of the constructed codes in case of column erasures as well as a column error. The algorithms work for both Construction \ref{cnstr1} and its duplication code.

Let $\cC$ be a $(k+2,k)$ MDS array code defined by Construction \ref{cnstr1} (and possibly duplication). The code has array size $2^m \times (k+2)$. Let the zigzag permutations be $f_j$, $j \in [0,k-1]$, which are not necessarily distinct.
Let the information elements be $a_{i,j}$, and the row and zigzag parity elements be $r_{i}$ and $z_{i}$, respectively, for $i\in [0,2^m-1],j \in [0,k-1]$.
W.l.o.g. assume the row coefficients are $\alpha_{i,j}=1$ for all $i,j$. And let the zigzag coefficients be $\beta_{i,j}$ in some finite field $\mathbb{F}$.

The following is a summary of the erasure decoding algorithms mentioned in the previous sections.
\begin{alg}(Erasure Decoding)\\
{\bf One erasure.} \\
1) One parity node is erased. Rebuild the row parity by
\begin{equation} \label{eq10}
r_{i}=\sum_{j=0}^{k-1}a_{i,j},
\end{equation}
and the zigzag parity by
\begin{equation} \label{eq11}
z_{i}=\sum_{j=0}^{k-1}\beta_{f_j^{-1}(i),j} a_{f_j^{-1}(i),j}.
\end{equation}
2) One information node $j$ is erased. Rebuild the elements in rows $X_j$ (see Construction \ref{cnstr1}) by rows, and those in rows $\overline{X_j}$ by zigzags. \\
{\bf Two erasures.} \\
1) Two parity nodes are erased. Rebuild by \eqref{eq10} and \eqref{eq11}.\\
2) One parity node and one information node is erased. If the row parity node is erased, rebuild by zigzags; otherwise rebuild by rows.\\
3) Two information node $j_1$ and $j_2$ are erased. \\
- If $f_{j_1}=f_{j_2}$, for any $i \in [0,2^m-1]$, compute
\begin{eqnarray} \label{eq12}
x_i&=&r_{i}-\sum_{j \neq j_1,j_2} a_{i,j} \nonumber \\
y_i&=&z_{f_{j_1}(i)} - \sum_{j \neq j_1,j_2}\beta_{f_j^{-1}f_{j_1}(i),j} a_{f_j^{-1}f_{j_1}(i),j}\nonumber
\end{eqnarray}
Solve $a_{i,j_1}, a_{i,j_2}$ from the equations
$$\left[
\begin{array}{ll}
	1 & 1 \\
	\beta_{i,j_1} & \beta_{i,j_2}
\end{array}
\right]
\left[
\begin{array}{l}
a_{i,j_1} \\
a_{i,j_2}
\end{array}
\right]
= \left[
\begin{array}{l}
x_i \\
y_i
\end{array}
\right].
$$
- Else, for any $i \in [0,2^m-1]$, set $i'=i+f_{j_1}(0)+f_{j_2}(0)$, and compute $x_i,x_{i'},y_i,y_{i'}$ according to \eqref{eq12}. Then solve $a_{i,j_1},a_{i,j_2},a_{i',j_1},a_{i',j_2}$ from equations
$$
\left[\begin{array}{cccc}
1 & 1 & 0 & 0 \\
0 & 0 & 1 & 1 \\
\beta_{i,j_1} & 0 & 0 & \beta_{i',j_2} \\
0 & \beta_{i,j_2} & \beta_{i',j_1} & 0
\end{array} \right]
\left[ \begin{array}{c}
a_{i,j_1} \\
a_{i,j_2} \\
a_{i',j_1} \\
a_{i',j_2} \\
\end{array} \right]
= \left[ \begin{array}{c}
x_i \\
x_{i'} \\
y_i \\
y_{i'}
\end{array} \right].
$$
\end{alg}

In case of a column error, we first compute the syndrome, then locate the error position, and at last correct the error.
Let $x_0,x_1,\dots,x_{p-1} \in \mathbb{F}$. Denote $f^{-1}(x_0,x_1,\dots,x_{p-1})=(x_{f^{-1}(0)},x_{f^{-1}(1)},\dots,x_{f^{-1}(p-1)})$ for a permutation $f$ on $[0,p-1]$.
The detailed algorithm is as follows.

\begin{alg} (Error Decoding) \\
 Compute for all $i \in [0,2^m-1]$:
\begin{eqnarray*}
S_{i,0}&=&\sum_{j=0}^{k-1}a_{i,j}-r_{i} \\
S_{i,1}&=&\sum_{j=0}^{k-1}\beta_{f_j^{-1}(i),j} a_{f_j^{-1}(i),j} -z_i
\end{eqnarray*}
Let the syndrome be $S_{0}=(S_{0,0},S_{1,0},\dots,S_{2^m-1,0})$ and $S_{1}=(S_{0,1},S_{1,1},\dots,S_{2^m-1,1})$. \\
- If $S_0=0$ and $S_1=0$, there is no error. \\
- Else if one of $S_0, S_1$ is $0$, there is an error in the parity. Correct it by \eqref{eq10} or \eqref{eq11}.\\
- Else, find the error location. For $j = 0$ to $k-1$: \\
\tab Compute for all $i \in [0,2^m-1]$, $x_{i} = \beta_{i,j}S_{i,0}.$ \\
\tab Let $X=(x_0,\dots,x_{2^m-1})$ and $Y=f_j^{-1}(X)$. \\
\tab If $Y=S_1$, subtract $S_0$ from column $j$. Stop.\\
If no such $j$ is found, there are more than one error.
\end{alg}

If there is only one error, the above algorithm is guaranteed to find the error location and correct it, since the code is MDS. If the computations are done in parallel for all $i \in [0,2^m-1]$, then this algorithm can be done in time $O(k)$. Moreover, since the permutations $f_i$'s only change one bit of a number in $[0,2^m-1]$ in Theorem \ref{orthogonal-permutations}, the algorithm can be easily implemented.

%
%
\section{Generalization of the code construction}
\label{generalization}
%
%
In this section we generalize Construction \ref{cnstr1} into arbitrary number of parity nodes. Let $n-k=r$ be the number of parity nodes, we will construct an $(n,k)$ MDS array code, i.e., it can recover up to $r$ nodes erasures for arbitrary integers $n,k$. We assume that each systematic node stores $\frac{\cM}{k}$ of the information and is stored in columns $[0,k-1]$. The $i$-th parity node is stored in column $k+i$, where $0\leq i\leq r-1$ and is comprised of zigzag sets $\{Z_{j}^i\}$.
\begin{cnstr}
\label{cnstr5}
Let $T$ be a subset of vector of $\mathbb{Z}_r^m$ for some integer $m$, where for each
\begin{equation}
v=(v_1,...,v_m)\in T,\gcd(v_1,...,v_m,r)=1.
\label{eq:45678}
\end{equation}
We will construct for each integer $i,0\leq i\leq r-1$ a set of permutations $B_T^i=\{f_v^i:v\in T\}$ where each permutation acts on the set $[0,r^m-1]$. By abuse of notations we use $x$ both to represent the integer and its $r$-ary representation and all the calculations are done over $\mathbb{Z}_r$. Define $f_v^i(x)=x+iv$, for example if $r=3, x=5, i=2, v=(0,1)$
$$f_{(0,1)}^2(5)=5+2(0,1)=(1,1)+(0,2)=(1,0)=3.$$
We define the zigzag $z_m^s$ in parity node $s,0\leq s \leq r-1$ and $m\in [0,r^m-1]$, as the linear combination of the elements $a_{i,j}$ such that their coordinates  satisfy
$$f^s_{v_j}(i)=m.$$
Note that these parity nodes described now, form an $(n,k)$ MDS array code under appropriate selection of the coefficients in the linear combinations of the zigzag. The proof follows the exact same lines as in Theorem \ref{zigzag-sets} and therefore is omitted. In a rebuilding of systematic node $i$ the elements in rows $X^s_i=\{v\in [0,r^m-1]:v\cdot v_i=r-s\}$ are rebuilt by parity node $s$ where $0\leq s\leq r-1$. From \eqref{eq:45678} we get that for any $i$ and $s$,
$|X^s_i|=r^{m-1}.$
\end{cnstr}

Let the information array be $A=(a_{i,j})$ with size $r^m \times |T|$.
Denote by $c(u)$ the column index corresponding to vector $u \in T$.
Define $f_u^{-i}(x) = x-iu$, which is the inverse permutation of $f_u^i$, $0 \le i \le r-1$.
For any $x \in [0,r^m-1]$, the $x$-th zigzag set in  the $i$-th parity node is $Z_x^i=\{a_{f_u^{-i}(x),c(u)}: u \in T\}.$
In a rebuilding of column $c(v)$, by Construction \ref{cnstr5}, we need to access the elements in the zigzag sets
$Z_x^i$ for any $x\in f_v^i(X_v^i)$ and $i\in [0,r-1]$. Hence, in column $c(u)$, $u \neq v$, we access elements in rows
\begin{equation}
\cup_{i=0}^{r-1} f_u^{-i} f_v^{i} (X_v^i).
\label{eq:1212}
\end{equation}
In order to get a low rebuilding ratio, we want that the union of \eqref{eq:1212} over all the columns $c(u),u\neq v$ to be as small as possible.

In particular, for a given $u$ if $f_u^{-i} f_v^{i}(X_v^i)$ are identical for any $0 \le i \le r-1$, we need only to access  $|f_u^{-i} f_v^{i}(X_v^i)|=r^{m-1}$ elements from column $c(u)$. Notice that in this case, for all $0 \le i \le r-1$,
\begin{equation} \label{orth:gen}
f_u^i(X_v^0) = f_v^i (X_v^i),
\end{equation}
simply because $f_v^0$ is the identity permutation for any vector $v$. We say the permutation set $\{f_u^i\}_{0 \le i \le r-1}$ is \emph{orthogonal} to the permutation set $\{f_{v}^{i}\}_{0 \le i \le r-1}$ if (\ref{orth:gen}) holds for all $i$.

In general, we want to know the relation of $f_u^{-i} f_v^{i}(X_v^i)$, $0 \le i \le r-1$. And the following lemma shows they are either identical or disjoint.

\begin{lem} \label{lem:orth}
For any $u,v$,
\begin{displaymath}
   |f_u^{-i} f_v^{i}(X_v^i) \cap f_u^{-j} f_v^{j}(X_v^j)|
=  \left\{
\begin{array}{ll}
|X_v^0|, & (i-j)c_{v,u}=0 \\
0,       & \text{o.w.}
\end{array}
\right.
\end{displaymath}
Where $c_{v,u}=v\cdot(v-u)-1$. In particular for $j=0$ we get
\begin{displaymath}
|f_u^i(X_v^0) \cap f_v^i (X_v^i)| = \left\{
\begin{array}{ll}
|X_v^0|, & \text{if } ic_{v,u}=0 \\
0,       & \text{o.w.}
\end{array}
\right.
\end{displaymath}
 \end{lem}

\begin{IEEEproof}
Consider the group $(\mathbb{Z}_r^m,+)$. Let $a_v^i\in X_v^i$ and $a_v^j\in X_v^j$. Note that $X_v^0$ is a subgroup of $\mathbb{Z}_r^m$ and $X_v^i=X_v^0+a_v^i$ is its coset. Hence the cosets $f_u^{-i} f_v^{i}(X_v^i)=X_v^0+a_v^i+i(v-u)$ and
$f_u^{-j} f_v^{j}(X_v^j)=X_v^0+a_v^j+j(v-u)$ are either identical or disjoint. Moreover they are identical if and only if
$$a_v^i-a_v^j+(i-j)(v-u)\in X_v^0,$$
i.e., $(i-j)\cdot c_{v,u}=0$ and the result follows.
\end{IEEEproof}

\begin{thm}
\label{mashumashu}
The permutations set $\{f_l^i\}$ and sets $X_l^i$, for $0 \le l \le m, 0 \le i \le r-1$ constructed by the vectors $\{e_i\}_{i=0}^m$ and Construction \ref{cnstr5} where $X_0^i$ is modified for any $i$ to be $X_0^i=\{x\in \mathbb{Z}_r^m:x\cdot (1,1,...,1)=i\}$, is a set of orthogonal permutations. Moreover the $(m+1+r,m+1)$ MDS array code of array size $r^m\times (m+1+r)$ defined by these permutations has \emph{optimal} ratio of $\frac{1}{r}.$
\end{thm}
\begin{IEEEproof}
For $1 \le l,n \le m$, $c_{l,n}=e_l \cdot(e_l-e_n)-1 = 0$, so by Theorem \ref{lem:orth}, $\{f_n^i\}_{0 \le i \le r-1}$ is orthogonal to $\{f_l^i\}_{0 \le i \le r-1}$. For $1 \le n \le m$, and all $0 \le i \le r-1$,
$$f_0^i(X_n^0) = X_n^0=\{x:x_n=0\}$$
and
\begin{eqnarray*}
f_n^i(X_n^i) &=& f_n^i(\{x:x_n=-i\})=\{x+ie_n:x_n=-i\} \\
&=&\{y:y_n=0\}
\end{eqnarray*}
Therefore, $\{f_0^i\}_{0 \le i \le r-1}$ is orthogonal to $\{f_n^i\}_{0 \le i \le r-1}$. Similarly,
$$f_n^i(X_0^0)=f_n^i(\{x:x\cdot(1,\dots,1)=0\})=\{x:x\cdot(1,\dots,1)=i\}$$
and
$$f_0^i(X_0^i)=X_0^i=\{x:x\cdot(1,\dots,1)=i\}$$
for all $i$. Thus, $\{f_n^i\}_{0 \le i \le r-1}$ is orthogonal to $\{f_0^i\}_{0 \le i \le r-1}$. Therefore, in order to rebuild one systematic node, we access only $r^{m-1}$ elements in each surviving node.
Hence the ratio is $1/r$.
\end{IEEEproof}

The following theorem gives the ratio for any code of Construction \ref{cnstr5}.

\begin{thm} \label{thm: gen rate}
The ratio for the code constructed by Construction \ref{cnstr5} and set of vectors $T$  is
$$\frac{\sum_{v \in T} \sum_{u \neq v \in T} \frac{1}{\gcd(r,c_{v,u})} + |T| }{|T|(|T|-1+r)}$$
where $\gcd$ is the greatest common divisor.
\end{thm}
\begin{IEEEproof}
By previous discussions and noticing that we access $r^{m-1}$ elements in each parity node, the ratio is
\begin{equation} \label{eq:union}
\frac{\sum_{v \in T} \sum_{u \neq v \in T} |\cup_{i=0}^{r-1} f_u^{-i} f_v^{i} (X_v^i)| + |T|r^m} {|T|(|T|-1+r) r^m}
\end{equation}
From Lemma \ref{lem:orth}, and noticing that $|\{i: ic_{v,u}=0 \mod r\}|=\gcd(r,c_{v,u})$, we get
$$|\cup_{i=0}^{r-1} f_u^{-i} f_v^{i} (X_v^i)|=r^{m-1}\times r /\gcd(r,c_{v,u}).$$
Thus, the theorem follows.
\end{IEEEproof}

When $c_{v,u}=0$ for all $v,u \in T$, the ratio is $1/r$ and is optimal according to (\ref{eq:tradeoff}). The next theorem states that the construction using standard basis has the most number of columns possible for an optimal ratio code.

\begin{thm}
For an orthogonal set of permutations $\{f_l^i\}_{0 \le l \le s-1, 0 \le i \le r-1}$ over the integers $[0,r^k-1]$, and the corresponding sets $\{X_l^i\}$, i.e., for all $0 \le m,l \le s-1,0 \le i \le r-1$,
$$f_m^i(X_l^0) = f_l^i (X_l^i)  $$
we have $s \le k+1$.
\end{thm}
\begin{IEEEproof}
We prove it by induction on $k$.
When $k=0$, $s \le 1$.
Notice that for permutations $g,h_i$, $0 \le i \le r-1$, the set of permutations $\{h_i f_l^i g\}$ are still orthogonal under the sets $\{g^{-1}(X_l^i)\}$. This is because
$h_i f_m^i g (g^{-1}(X_l^0)) = h_i f_m^i(X_l^0) = h_i f_l^i (X_l^i) = h_i f_l^i g (g^{-1}(X_l^i))$.
Therefore, we can assume w.l.o.g. $X_0^{i}=[ir^{k-1},(i+1)r^{k-1}-1]$, and $f_0^{i}$ is identity, for $0 \le i \le r-1$.

Since $f_l^i(X_0^0)=f_0^i(X_0^i)=X_0^i$ for all $l \neq 0$, we have
$f_m^i(X_l^0 \cap X_0^0), f_l^i(X_l^i \cap X_0^0)$ $\subseteq$ $X_0^i$.
Similarly, for all $l \neq 0$,
$f_l^i(\overline{X_0^0})= \overline{f_l^i(X_0^0)} =\overline{X_0^i}$ and
$f_m^i(X_l^0 \cap \overline{X_0^0}), f_l^i(X_l^i \cap \overline{X_0^0})$ $\subseteq$ $\overline{X_0^i}$.
But $f_m^i(X_l^0) = f_l^i (X_l^i)$, which implies
$$f_m^i(X_l^0 \cap X_0^0)=f_l^i(X_l^i \cap X_0^0).$$
Thus the set of permutations $\{\hat{f}_l^i\}_{1 \le l \le s-1, 0 \le i \le r-1}$ over $[0,r^{k-1}-1]$, where $\hat{f}_l^i(x)=x-ir^{k-1}$  and sets $\{X_l^i \cap X_0^0\}$ are orthogonal. And by induction, $s-1 \le k$.
\end{IEEEproof}

The next theorem gives the finite field size of a code with $3$ parities.

\begin{thm}
For $r=3$ parity nodes the code constructed by the standard basis and the zero vector$\{e_i\}_{i=0}^m\in \mathbb{F}_3^m$ together with Theorem \ref{mashumashu}, a field of size at most $2(m+1)$ suffices to make it an $(m+4,m+1)$ MDS code.
\end{thm}
\begin{proof}
Let $\mathbb{F}_q$ be a field of size $q\geq 2(m+1)$, and $q$ is a prime. For any  $l\in [0,m]$ let $A_l=(a_{i,j}),$ be the permutation matrix defined by the permutation $f_{e_l}^1$ i.e. $a_{i,j}=1$ iff
$f_{e_l}^1(j)=i$, else $a_{i,j}=0$. Let $a$ be a primitive element of the field $\mathbb{F}_q$, and modify the nonzero entries of $A_l$ as follows, if $a_{i,j}=1$ and $j \cdot e_l=0$, modify it to $a_{i,j}=a^l.$

We will show that under this assignment of coefficients the matrices $A_l$ commute, i.e. for any $l_1 \neq l_2\in [0,m],A_{l_1}A_{l_2}=A_{l_2}A_{l_1}$. For simplicity, write $f_{e_{l_1}}=f_1,f_{e_{l_2}}=f_2$ and $p=3^m$. And write $A_{l_1}=(a_{i,j}),A_{l_2}=(b_{i,j})$.
For a vector $x=(x_0,\dots,x_{p-1})$ and $y = x A_{l_1}$,
its $j$-th entry satisfies $y_j=a_{f_1(j),j}x_{f_1(j)}$ for all $j \in [0,p-1]$. And by similar calculation, $z=x A_{l_1} A_{l_2} = y A_{l_2}$ will satisfy
$$z_j = b_{f_2(j),j}y_{f_2(j)}=
b_{f_2(j),j} a_{f_1(f_2(j)),f_2(j)}x_{f_1(f_2(j))}.$$
Similarly, if $w=x A_{l_2} A_{l_1}$, then
$$w_j = a_{f_1(j),j} b_{f_2(f_1(j)),f_1(j)}x_{f_2(f_1(j))}.$$
Notice that $$ f_1(j) \cdot e_{l_2} = (j+e_{l_1})e_{l_2}=j \cdot e_{l_2},$$
so  $b_{f_2(j),j}=b_{f_2(f_1(j)),f_1(j)}$. Similarly, $a_{f_1(j),j}=a_{f_1(f_2(j)),f_2(j)}$. Moreover, $$f_1(f_2(j))=f_2(f_1(j))=j+e_{l_1}+e_{l_2}.$$ Hence, $z_j=w_j$ for all $j$ and
$$x A_{l_1} A_{l_2} =z=w= x A_{l_2} A_{l_1}$$
for all $x \in \mathbb{F}_3^{m}$. Thus $A_{l_1}A_{l_2}=A_{l_2}A_{l_1}$.

Note that $f_{e_l}^2=(f_{e_l})^2$ hence define the set of matrices $A_l^2$ to be $A_l^2=(A_l)^2$. The code is MDS if it can recover from loss of any $3$ nodes. Hence With this assignment of coefficients the code is MDS iff any block submatrices of sizes $1\times1,2\times 2,3\times 3$ of the matrix
$$\left[
\begin{array}{llll}
	I & I &...& I\\
	A_0 & A_1 &...& A_m \\
	A_0^2 & A_1^2 &...& A_m^2
\end{array}
\right]
$$
are invertible. Let $0\leq i<j<k\leq m$ we will see that the matrix
$$\left[
\begin{array}{lll}
	I & I & I\\
	A_{i} & A_j & A_k \\
	A_{i}^2 & A_j^2 & A_k^2
\end{array}
\right]
$$
is invertible. By Theorem $1$ in \cite{Determinants} and the fact that  all the blocks in the matrix commute we get that the determinant equals to $ \det(A_k-A_j)\cdot\det(A_k-A_i)\cdot\det(A_j-A_i)$. Hence we need to show that for any $i>j$ that $\det(A_i-A_j)\neq 0$ which is equivalent to $\det(A_iA_j^{-1}-I)\neq 0.$ Note that for any $i$, $A_i^3=a^iI$. Denote by $A=A_iA_j^{-1}$, hence $A^3=(A_iA_j^{-1})^3=A_i^3A_j^{-3}=a^{i-j}I\neq I$. Therefore
$$0\neq\det(A^3-I)=\det(A-I)\det(A^2+A+I).$$
Therefore $\det(A-I)=\det(A_iA_j^{-1}-I)\neq 0$.
For the case when one of the three erased columns is a parity column it sufficient to check that for any $i>j$ that
$$\det(\left[
\begin{array}{ll}
	I & I \\
	A_j^2 & A_i^2
\end{array}
\right])=\det(A_j^{-2})\det(A_i^2A_j^{-2}-I)\neq 0.
$$
Note that $A^6=(A_iA_j^{-1})^6= a^{2(i-j)}I\neq I$ since $i-j\leq m< \frac{q-1}{2}.$ Hence
$$0\neq \det(A^6-I)=\det(A^2-I)(A^4+A^2+I),$$ and $\det(A^2-I)=\det(A_i^2A_j^{-2}-I)\neq 0$ which concludes the proof.
\end{proof}

%
\section{Concluding Remarks}
\label{summary}
%
%
In this paper, we described explicit constructions of the first known systematic $(n,k)$ MDS array codes with $n-k$ a constant (independent of $n$) and with an optimal rebuilding ratio. Specifically, the amount of information needed to rebuild an erased column equals to $1/(n-k)$, which matches the information-theoretic lower bound.

Here we provide additional observations and indicate a couple of open research problems.
 
\emph{The write complexity:} Consider the number of read accesses in the case of a write (update) operation.  For example, in an array code with two redundancies, in order to update a single information element, one needs to read at least three times and write three times. The reason is that in order to compute the new parity element, we need to know the values of the current information element and the two current parity elements. However, in our optimal code construction with two redundancies, if we update all the information in column 1 and the rows in the first half of the array (see Figure \ref{fig2}), it is not necessary to read any information or parity elements, since we know the values of all the information elements needed for computing the parity. These information elements take about half the size of the entire array - the idea is to cache the information to be written until most of the corresponding elements need update (we could arrange the information in a way that these elements are often updated at the same time), and the number of reading operations compared to the information size is very small. Similarly we can use the same approach for any other systematic column. In general, given $r$ redundancies, we could avoid read operations (in support of a write operation) if we make sure that we update about $1/r$ of the array at the same time.

\emph{Lowering the rebuilding ratio:} We note that one can add redundancy columns for the sake of lowering the rebuilding ratio. For instance, one can use three redundancy columns where the third column is not used for erasure correction - with three redundancy columns we need to access only $1/3$ of data instead of $1/2$ in the rebuilding of a single failed column.

\emph{Open problem:}
For a given array size, what is the rebuilding ratio for a code defined by arbitrary permutations? In Theorem \ref{thm:opt rate}, we showed that $1/2+1/k$ is optimal for a code constructed by binary vectors and duplication. However, this ratio is not known for arbitrary permutations.


\end{document}